\documentclass[aps, prl, twocolumn]{revtex4-1} 
\usepackage{graphicx}
\usepackage{amsmath}
\usepackage{amssymb}
\usepackage{amsthm}
\usepackage{tikz}
\usetikzlibrary{decorations.markings,decorations.pathreplacing,calc}
\usepackage{algorithm}
\usepackage{algpseudocode}
\usepackage{thm-restate}

\usepackage{subfig}

\newcommand{\be}{\begin{equation}}
\newcommand{\ee}{\end{equation}}
\newcommand{\ba}{\begin{array}}
\newcommand{\ea}{\end{array}}
\newcommand{\bea}{\begin{eqnarray}}
\newcommand{\eea}{\end{eqnarray}}

\newcommand{\trace}[1]{{\mathrm{Tr}{#1}}}

\newtheorem{prop}{Proposition}
\newtheorem{lemma}{Lemma}
\newtheorem{corol}{Corollary}

\newtheorem{theorem}{Theorem}
\newtheorem*{theorem*}{Theorem}

\begin{document}
\title{Polynomial-time classical simulation of quantum ferromagnets}
\author{Sergey Bravyi}
\author{David Gosset}
\affiliation{IBM T.J. Watson Research Center}
\begin{abstract}We consider a family of quantum spin systems which includes as special cases the ferromagnetic XY model and ferromagnetic Ising model on any graph, with or without a transverse magnetic field. We prove that the partition function of any model in this family can be efficiently approximated to a given relative error $\epsilon$ using a classical randomized algorithm with runtime polynomial in $\epsilon^{-1}$, system size, and inverse temperature. As a consequence we obtain a polynomial time algorithm which approximates the free energy or ground energy to a given additive error. We first show how to approximate the partition function by the perfect matching sum of a finite graph with positive edge weights. Although the perfect matching sum is not known to be efficiently approximable in general, the graphs obtained by our method have a special structure which facilitates efficient approximation via a randomized algorithm due to Jerrum and Sinclair. 
\end{abstract}
\maketitle

Quantum Monte Carlo is an umbrella term which refers to a powerful suite of classical probabilistic simulation algorithms for quantum many-body systems.  These algorithms can be used to compute the thermal or ground state properties of a quantum system described by a \textit{stoquastic} (sign-problem free) Hamiltonian, defined by the property that all off-diagonal matrix elements are real and nonpositive.  
In practice quantum Monte Carlo methods can be used to simulate systems which are orders of magnitude larger than is possible using exact diagonalization~\cite{sandvikreview}. This is because they are based on a probabilistic representation of the Gibbs density matrix which typically uses substantially less computer memory than an explicit representation.

Given the empirical success of quantum Monte Carlo, one may ask if stoquasticity makes classical simulation easier in a formal complexity-theoretic sense. This and related questions have been studied in Refs.~\cite{stoq06,CM16,BH16, BT09,H13,CrossonHarrow}. For our purposes suffice it to say that stoquasticity alone is not enough to guarantee efficient simulation. Indeed, it is well known that approximating the ground energy is intractable even for the special case of classical spin Hamiltonians such as the Ising model \cite{barahona}.

Can we identify physically motivated families of stoquastic Hamiltonians for which efficient simulation \textit{is} possible? In a landmark result Jerrum and Sinclair established that a broad family of classical Ising models characterized by ferromagnetic interactions can be efficiently simulated \cite{JSising} (see also \cite{GJ15}). A recent extension efficiently simulates the ferromagnetic tranverse field Ising model \cite{B14}, a system described by a fully quantum (i.e., non-diagonal) stoquastic Hamiltonian. This result can be viewed as a provably efficient quantum Monte Carlo algorithm.

In this paper we consider a more general family of ferromagnetic quantum spin systems described by $n$-qubit Hamiltonians of the form
\begin{equation}
H=\sum_{1\leq i<j\leq n} \left(-b_{ij} X_i X_j+c_{ij} Y_i Y_j\right)+\sum_{i=1}^{n} d_i (I+Z_i).
\label{eq:ham}
\end{equation}
Here $b_{ij},c_{ij},d_i\in \mathbb{R}$ are some coefficients, and $X_i,Y_i,Z_i$ are Pauli operators acting on the $i$th qubit. We restrict our attention to ferromagnetic interactions, defined as:
\begin{flalign}
& \quad \textbf{Ferromagnetic}: \qquad |c_{ij}|\leq b_{ij}.&
\label{eq:ferro}
\end{flalign}
Eq.~\eqref{eq:ferro} ensures that $H$ is stoquastic. We also assume $b_{ij},|c_{ij}|,|d_i|\in [0,1]$, which can always be achieved by rescaling the Hamiltonian. 

We describe a polynomial-time classical approximation algorithm for the partition function
\begin{align}
\mathcal{Z}(\beta,H)&\equiv\mathrm{Tr}\left[e^{-\beta H}\right] \qquad \quad \beta>0,
\end{align}
which enables an efficient computation of the free energy and ground energy of these models, in a precise sense detailed below. We obtain efficient simulations of well-known models of ferromagnetism such as the XY model (setting $c_{ij}=-b_{ij})$, the transverse Ising model~\cite{B14} (setting $c_{ij}=0$), as well as a continuum of systems in between.  Our work further resolves the boundary between easy- and hard-to-simulate systems, a topic of central interest in quantum Hamiltonian complexity \cite{QHCreview}. As an example, we note that our results settle the complexity of a one-parameter family of local Hamiltonian problems studied by Piddock and Montanaro (see Fig. 1 of Ref.~\cite{PM15}). Our algorithm, which is provably efficient yet far from practical, complements existing quantum Monte Carlo methods which may be practical but have no performance guarantees \cite{sandvikreview}.

We say that $f\in \mathbb{R}$ approximates $F\in \mathbb{R}$ within relative error $\epsilon$ if $(1-\epsilon)F\leq f\leq (1+\epsilon)F$. A \textit{randomized approximation scheme} for a real-valued function $F$ with domain $\mathcal{D}$ is a classical probabilistic algorithm which takes as input $x\in \mathcal{D}$ and $0<\epsilon<1$ and outputs an estimate $f(x)$ which, with probability at least $3/4$, approximates $F(x)$ within relative error $\epsilon$. Here the success probability $3/4$ can always be amplified to $1-\delta$ by taking the median of  $O(\log(\delta^{-1}))$ estimates produced by independent runs of the algorithm \cite{JVV86}. 

We now state our main result.
\begin{theorem}
The partition function $\mathcal{Z}(\beta,H)$ admits a randomized approximation scheme with runtime upper bounded as $\mathrm{poly}(n,\beta,\epsilon^{-1})$.
\label{thm:main}
\end{theorem}
A direct application of Theorem \ref{thm:main} provides an estimate of the free energy $\mathcal{F}(\beta)=-\beta^{-1}\log(\mathcal{Z}(\beta,H))$ which, with high probability, achieves a given absolute error $\Delta$ in time $\mathrm{poly}(n,\beta,\Delta^{-1})$. Choosing $\beta=O(n\Delta^{-1})$ is sufficient  to ensure that the free energy approximates the ground energy to within absolute error $\Delta$ \footnote{Changing variables to temperature $T=\beta^{-1}$ , the ground energy is $\mathcal{F}(T=0)=E_0$ and $\frac{d\mathcal{F}}{dT}=-S(T)$ where $S(T)$ is the von Neumann entropy of the Gibbs density matrix. Since $0\leq S(T)\leq n$ we have $0\leq E_0-\mathcal{F}(T)\leq n T$.}.

To prove Theorem \ref{thm:main} we first approximate the partition function by the perfect matching sum of a weighted graph, which is then approximated using an algorithm from Ref.~\cite{JS89}.

We begin by reviewing relevant background information concerning matchings in graphs. We consider finite weighted graphs $\Gamma=(V,E,w)$ with vertex set $V$, edge set $E$ and positive edge weights, i.e., $w:E\rightarrow \mathbb{R}_{>0}$. A matching of $\Gamma$ is a set of edges $M\subseteq E$ such that no two edges in $M$ share a vertex. A perfect matching has the additional property that each vertex $v\in E$ is covered by exactly one edge of $M$. Equivalently, a matching is perfect if it contains exactly $|V|/2$ edges.  A nearly perfect matching is a matching with exactly $|V|/2-1$ edges.

Suppose $\Gamma$ has an even number $|V|=2N$ of vertices.  Let $M_{k}(\Gamma)$ be the set of all matchings of $\Gamma$ which contain exactly $k$ edges. Define the perfect matching sum
\begin{equation}
\mathrm{PerfMatch}(\Gamma)=\sum_{M\in M_N(\Gamma)} \prod_{e\in M} w(e) .
\label{eq:pmatch}
\end{equation}
Similarly define $\mathrm{NearPerfMatch}(\Gamma)$ to be Eq.~\eqref{eq:pmatch} with $M_N(\Gamma)$ replaced by $M_{N-1}(\Gamma)$.

The problem of computing $\mathrm{PerfMatch}(\Gamma)$ for a graph with positive edge weights has been extensively studied. If $\Gamma$ is planar it can be computed efficiently using the so-called FKT algorithm \cite{Fisher1961,Kasteleyn2009,Temperley1961}.  In general it is $\#$P-hard to compute exactly, even for bipartite graphs \cite{valiant79}. Nevertheless, for bipartite graphs it admits a randomized approximation scheme with runtime upper bounded as a polynomial in the size of the graph $|V|$ and the desired relative error $\epsilon^{-1}$ \cite{JSV04}. For general graphs (which may be neither planar nor bipartite), a randomized approximation scheme from Ref.~\cite{JS89} has polynomial runtime under the additional condition that $\mathrm{NearPerfMatch}(\Gamma)$ exceeds $\mathrm{PerfMatch}(\Gamma)$ by at most a polynomial factor. To state the result precisely, let $w_{max}=\max\left\{1,\max_{e\in E} w(e)\right\}$ and  $w_{min}=\min\left\{1, \min_{e\in E} w(e)\right\}$.

\begin{restatable}[Jerrum and Sinclair \cite{JS89}]{theorem}{hafnian}
Let $q$ be a fixed polynomial and let $\Gamma=(V,E,w)$ be a graph with positive edge weights such that $\mathrm{PerfMatch}(\Gamma)\neq 0$ and 
\begin{equation}
\frac{\mathrm{NearPerfMatch}(\Gamma)}{\mathrm{PerfMatch}(\Gamma)}\leq q(|V|).
\label{eq:pmcond}
\end{equation}
Then $\mathrm{PerfMatch}(\Gamma)$ admits a randomized approximation scheme with runtime bounded as $\tilde{O}\left(\epsilon^{-2}|V|^6|E|^5w_{max}^6 q(|V|)^6\right)$.
\label{thm:countpm}
\end{restatable}
Here the $\tilde{O}(\cdot)$ notation hides factors polynomial in $\log(\epsilon^{-1}),\log(w_{max}/w_{min})$ and $\log(|V|)$. A proof of Theorem \ref{thm:countpm} for unweighted graphs is given in Ref.~\cite{JS89}. In the Supplementary Material we adapt the proof (with superficial modifications) to establish Theorem \ref{thm:countpm}. 

Let us informally sketch of the proof of Theorem 1.  First, we use a quantum-to-classical  mapping based on the Trotter-Suzuki expansion to approximate  the quantum partition function $\mathcal{Z}(\beta,H)$ by the matching sum $\mathrm{PerfMatch}(\Gamma)$ for a suitable weighted graph $\Gamma$. The size of this graph scales polynomially in $n$, $\beta$, and $\epsilon^{-1}$. Thus it suffices to show that $\mathrm{PerfMatch}(\Gamma)$ can be approximated efficiently using the algorithm from Theorem \ref{thm:countpm}. Although in general this algorithm is not efficient, weighted graphs obtained via the quantum-to-classical mapping have a very special structure. This structure allows us to relate the ratio of two matching sums in Eq.~\eqref{eq:pmcond} to certain physical properties of the original quantum model. Loosely speaking, the ratio in Eq.~\eqref{eq:pmcond} can be expressed as a sum of imaginary time spin-spin correlation functions.  We obtain a constant upper bound on such correlation functions by examining  truncated versions of the Trotter-Suzuki expansion and using simple linear algebra arguments. Once a polynomial upper bound on the ratio Eq.~\eqref{eq:pmcond} is established,  Theorem \ref{thm:main} is obtained as a simple corollary of Theorem \ref{thm:countpm}.

We now proceed with the proof of Theorem \ref{thm:main}. Let the Hamiltonian $H$, inverse temperature $\beta>0$, and desired relative error $0<\epsilon<1$ be given. The first step is to establish an approximation
\begin{equation}
e^{-\beta H}\approxeq G_J\ldots G_2G_1,
\label{eq:prodZ}
\end{equation}
where each elementary gate $G_t$ is from a gate set $\mathcal{G}$ defined as follows. Define a one-qubit gate
\[
f(t)=\left({\begin{array}{cc} t & 0 \\ 0 & 1 \end{array}}\right)
\]
and two-qubit gates
\begin{equation}
g(t)=\left({\begin{array}{cccc} 1+t^2 & 0 & 0 & t\\ 0 & 1 & 0 & 0\\ 0 & 0 & 1 & 0\\ t & 0 & 0 & 1\end{array}}\right) \;\;\; h(t)=\left({\begin{array}{cccc} 1& 0 & 0 & 0\\ 0 & 1+t^2  & t & 0\\ 0 & t & 1 & 0\\ 0 & 0 & 0 & 1\end{array}}\right).
\label{eq:twoqubitgates}
\end{equation}
Here $t>0$ is a parameter. Define an $n$-qubit gate set 
\begin{equation}
\mathcal{G}=\left\{ f_{i}(2t),g_{ij}(t), h_{ij}(t) \; \bigg| \; i,j\in [n], \; i\neq j,\; 0<t< 1\right\}.
\label{eq:Gndef}
\end{equation}
where the subscripts indicate the qubit(s) on which the gate acts nontrivially.

To obtain Eq.~\eqref{eq:prodZ} we build $e^{-\beta H}$ as a product of exponentials of local terms simply related to the ones appearing in Eq.~\eqref{eq:ham}, \`{a} la Trotter-Suzuki. We use 
\begin{align}
f_i(e^{\pm 2s})&=e^{\pm s(I+Z_i)}\label{eq:fexp1}\\
g_{ij}(2s)&=e^{-s(Y_iY_j-X_i X_j)+O(s^2)}\label{eq:gexp1}\\
h_{ij}(2s)&=e^{-s(-Y_i Y_j-X_i X_j)+O(s^2)}\label{eq:hexp1},
\end{align}
and the fact that, due to Eq.~\eqref{eq:ferro}, each term $-b_{ij}X_iX_j+c_{ij}Y_iY_j$ in Eq.~\eqref{eq:ham} can be written as a linear combination of $Y_iY_j-X_i X_j$ and $-Y_i Y_j-X_i X_j$ with nonnegative coefficients. Using these ideas we establish an approximation with the following properties (the proof is given in the Supplementary Material). Here $H,\beta,\epsilon$ are as above.

\begin{lemma}[Trotter-Suzuki approximation]
We may choose $J=O((1+\beta^2) n^5\epsilon^{-1})$ and a sequence $G_1,G_2,\ldots,G_J\in \mathcal{G}$ with the following properties. There exists a Hermitian $Q$ with
\begin{equation}
G_J\ldots G_2G_1=e^{-\beta H+Q} \qquad  \|Q\|\leq \epsilon/4.
\label{eq:Gprod}
\end{equation}
Furthermore, for any $1\leq i<j\leq J$ there exists a Hermitian $W_{ij}$ and (possibly non-Hermitian) $R_{ij},L_{ij}$ with
\begin{equation}
G_{j} G_{j-1}\ldots G_i=L_{ij} e^{-\beta (j-i+1)J^{-1} H+W_{ij}}R_{ij}
\label{eq:Gpartialprod}
\end{equation}
where $\|W_{ij}\|\leq \epsilon/4$ and $\|R_{ij}\|,\|L_{ij}\|\leq 2$. In addition we have $R_{1j}=I$ for all $j>1$ and $L_{iJ}=I$ for all $i<J$.
\label{lem:pathint}
\end{lemma}
The gate sequence from the Lemma is a simple function of $H$, $\beta$ and $\epsilon$ and in particular is efficiently computable. Using this gate sequence we define
\[
\mathcal{Z}_J\equiv\mathrm{Tr}\left[G_J\ldots G_2G_1\right].
\]
From Weyl's inequality and Eq.~\eqref{eq:Gprod} we get
\begin{equation}
\mathcal{Z}(\beta,H)e^{-\epsilon/4}\leq \mathcal{Z}_J\leq \mathcal{Z}(\beta,H)e^{\epsilon/4}.
\label{eq:Zmapprox}
\end{equation}

In the second step of the proof we will show that 
\begin{equation}
\mathcal{Z}_J=\mathrm{Tr}[G_J G_{J-1}\ldots G_1]=\mathrm{PerfMatch}(\Gamma).
\label{eq:gamma}
\end{equation}
for a graph $\Gamma=(V,E,w)$ with non-negative edge weights. Here $\Gamma$ is related to the sequence $G_1,G_2,\ldots,G_J$ in a simple way described below. Moreover, it satisfies $|V|=O(J)$, $|E|=O(J)$, has maximum edge weight $w_{max}\leq 2$, and satisfies Eq.~\eqref{eq:pmcond} with $q(|V|)=O(|V|^2)$. We can therefore use the randomized approximation scheme from Theorem \ref{thm:countpm}  to compute an estimate $\mathcal{Z}_{\mathrm{est}}$ which, with probability at least $3/4$, approximates $\mathcal{Z}_J$ within relative error $\epsilon/2$. This estimate is computed using runtime
\[
\tilde{O}\left(J^{23}\epsilon^{-2}\right)=\tilde{O}\left(n^{115}(1+\beta^{46})\epsilon^{-25}\right)=\mathrm{poly}(n,\beta,\epsilon^{-1})
\]
Using Eq.~\eqref{eq:Zmapprox} and the fact that $e^{\epsilon/4}(1+\epsilon/2)\leq 1+\epsilon$ and $e^{-\epsilon/4}(1-\epsilon/2)\geq 1-\epsilon$ for all $\epsilon\in (0,1)$, we see that our estimate $\mathcal{Z}_{\mathrm{est}}$ approximates $\mathcal{Z}(\beta,H)$ within relative error $\epsilon$, with probability at least $3/4$. The estimate $\mathcal{Z}_\mathrm{est}$ is the output of our randomized approximation scheme, and the above shows that it satisfies the required error bound and is computed using the claimed polynomial runtime. 

It remains to construct $\Gamma$ and establish the properties stated above. 
\begin{figure}
\centering
\subfloat[Gadget for $f(t)$]{
\begin{tikzpicture}[scale=0.6]
\node at (0,-0.05){};
\draw (0,0.5)--(1,0.5);
\draw node[circle,fill=black, inner sep=1.5pt] at (0,0.5)(){};
\draw node[circle,fill=black, inner sep=1.5pt] at (1,0.5)(){};
\draw[thick,dotted] (-0.7,0.5)--(0,0.5);
\draw[thick,dotted] (1.7,0.5)--(1,0.5);
\draw node at (0.5,1){$t$};

\end{tikzpicture}
}
\hspace{0.5cm}
\subfloat[Gadget for $g(t)$]{
\begin{tikzpicture}[scale=0.6]

\node at (0,-0.1){};
\draw node at (-0.2,0.5){$t$};
\draw node at (1.2,0.5){$t$};
\draw (0,0)--(0,1)--(1,1)--(1,0)--(0,0);
\draw node[circle,fill=black, inner sep=1.5pt] at (0,0)(){};
\draw node[circle,fill=black, inner sep=1.5pt] at (0,1)(){};
\draw node[circle,fill=black, inner sep=1.5pt] at (1,0)(){};
\draw node[circle,fill=black, inner sep=1.5pt] at (1,1)(){};
\draw[thick,dotted] (-0.7,0)--(0,0);
\draw[thick,dotted] (1.7,0)--(1,0);
\draw[thick,dotted] (1,1)--(1.7,1);
\draw[thick,dotted] (0,1)--(-0.7,1);
\end{tikzpicture}
}
\hspace{0.5cm}
\subfloat[Gadget for $h(t)$]{
\begin{tikzpicture}[scale=0.6]

\node at (0,-0.1){};
\draw node at (-0.2,0.5){$t$};
\draw node at (1.2,0.5){$t$};
\draw (0,0)--(0,1)--(1,1)--(1,0)--(0,0);
\draw (0,0)--(-1,0);
\draw (1,0)--(2,0);

\draw node[circle,fill=black, inner sep=1.5pt] at (-1,0)(){};
\draw node[circle,fill=black, inner sep=1.5pt] at (2,0)(){};
\draw node[circle,fill=black, inner sep=1.5pt] at (0,0)(){};
\draw node[circle,fill=black, inner sep=1.5pt] at (0,1)(){};
\draw node[circle,fill=black, inner sep=1.5pt] at (1,0)(){};
\draw node[circle,fill=black, inner sep=1.5pt] at (1,1)(){};

\draw[thick,dotted] (-1.7,1)--(0,1);

\draw[thick,dotted] (2.7,1)--(1,1);

\draw[thick,dotted] (-1.7,0)--(-1,0);
\draw[thick,dotted] (2.7,0)--(2,0);
\end{tikzpicture}
}
\caption{Gadgets for elementary gates\label{fig:gadgets}}
\end{figure}
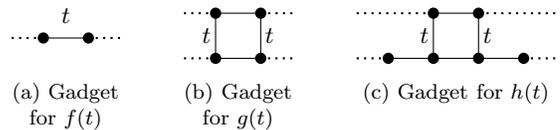
We build $\Gamma$ using a set of gadgets shown in Fig.~\ref{fig:gadgets}. Each gadget is a finite weighted graph along with a subset of distinguished vertices to which \textit{external edges} may be added, as shown by dotted lines in the Figure. Each distinguished vertex is either an input vertex (external edges attach on the left) or an output vertex (on the right). In the Figure, nontrivial edge weights are labeled whereas unlabeled edges are assigned weight $1$.

Consider a gadget with graph $\Theta=(V,E,w)$, input vertices $\{v_1,v_2,\ldots,v_m\}\subseteq V$ and output vertices $\{v_{m+1},v_{m+2},\ldots v_{2m}\}\subseteq V$. For each $x\in \{0,1\}^{2m}$ define $\Theta_x$ to be the induced subgraph obtained from $\Theta$ by removing all vertices $v_i$ such that $x_i=1$. An $m$-qubit gate $G$ associated to the gadget is defined via
\[
\langle x_{2m}\ldots x_{m+2}x_{m+1} |G |x_{m}\ldots x_2 x_1\rangle=\mathrm{PerfMatch}(\Theta_x).
\]
We say the gadget \textit{implements} the gate $G$. It is easily checked that the gadgets shown in Fig.~\ref{fig:gadgets}(a),(b), and (c) implement the gates $f(t), g(t)$, and $h(t)$ respectively.

The above definition composes nicely and allows us to form gadgets which implement the product or trace of a product of gates. Indeed, we can view the gate implemented by a gadget as a tensor with an index for each distinguished vertex. Adding a (weight 1) external edge between two distinguished vertices has the effect of contracting the corresponding indices. For example, if we add external edges which connect the output vertices of a gadget which implements a gate $G_1$ to the input vertices of a gadget which implements another gate $G_2$, we obtain a gadget which implements the product $G_2G_1$. 

We now construct the graph $\Gamma$ satisfying Eq.~\eqref{eq:gamma} (it closely resembles a circuit diagram). Start with a gadget from Fig.~\eqref{fig:gadgets} for each gate in the circuit. Draw the corresponding $J$ disjoint graphs in order $G_1,G_2,\ldots,G_J$ from left to right. Next add a set of weight-$1$ external edges as follows. Each qubit $j=1,2,\ldots,n$ is acted on by some number $1 \leq m_j\leq J$ of gates in the circuit, and is therefore associated with $m_j$ input vertices $\{\mathrm{in}^{(j)}_1,\mathrm{in}^{(j)}_2,\ldots \mathrm{in}^{(j)}_{m_j}\}$ and $m_j$ output vertices $\{\mathrm{out}^{(j)}_1,\mathrm{out}^{(j)}_2,\ldots \mathrm{out}^{(j)}_{m_j}\}$. These appear in alternating order $\mathrm{in}^{(j)}_1,\mathrm{out}^{(j)}_1,\mathrm{in}^{(j)}_2,\mathrm{out}^{(j)}_2,\ldots, \mathrm{in}^{(j)}_{m_j},\mathrm{out}^{(j)}_{m_j}$ from left to right. For each qubit $j\in \{1,2,\ldots,n\}$, we add $m_j$ weight $1$ edges 
\[
\left\{(\mathrm{out}^{(j)}_{m_j},\mathrm{in}^{(j)}_1)  \text{  and  } (\mathrm{out}^{(j)}_k,\mathrm{in}^{(j)}_{k+1}): 1\leq k\leq m_j-1 \right\}.
\]
\begin{figure}
\centering
\begin{tikzpicture}[scale=0.6]
\draw node at (-0.2,0.5){$a$};
\draw node at (1.2,0.5){$a$};
\draw (0,0)--(0,1)--(1,1)--(1,0)--(0,0);
\draw node[circle,fill=black, inner sep=1.5pt] at (0,0)(){};
\draw node[circle,fill=black, inner sep=1.5pt] at (0,1)(){};
\draw node[circle,fill=black, inner sep=1.5pt] at (1,0)(){};
\draw node[circle,fill=black, inner sep=1.5pt] at (1,1)(){};

\begin{scope}[shift={(2,-1)}]
\node at (0,-0.1){};
\draw node at (-0.2,0.5){$b$};
\draw node at (1.2,0.5){$b$};
\draw (0,0)--(0,1)--(1,1)--(1,0)--(0,0);
\draw node[circle,fill=black, inner sep=1.5pt] at (0,0)(){};
\draw node[circle,fill=black, inner sep=1.5pt] at (0,1)(){};
\draw node[circle,fill=black, inner sep=1.5pt] at (1,0)(){};
\draw node[circle,fill=black, inner sep=1.5pt] at (1,1)(){};
\end{scope}
\begin{scope}[shift={(5,-1)}]
\draw node at (-0.2,0.5){$c$};
\draw node at (1.2,0.5){$c$};
\draw (0,0)--(0,2)--(1,2)--(1,0)--(0,0);
\draw (0,0)--(-1,0);
\draw (1,0)--(2,0);

\draw node[circle,fill=black, inner sep=1.5pt] at (-1,0)(){};
\draw node[circle,fill=black, inner sep=1.5pt] at (2,0)(){};
\draw node[circle,fill=black, inner sep=1.5pt] at (0,0)(){};
\draw node[circle,fill=black, inner sep=1.5pt] at (0,2)(){};
\draw node[circle,fill=black, inner sep=1.5pt] at (1,0)(){};
\draw node[circle,fill=black, inner sep=1.5pt] at (1,2)(){};
\end{scope}

\draw node[circle,fill=black, inner sep=1.5pt] at (8,0)(){};
\draw node[circle,fill=black, inner sep=1.5pt] at (7,0)(){};
\draw (7,0)--(8,0);
\node at (7.5,0.3){$d$};

\draw (0,0)--(7,0);
\draw (3,-1)--(4,-1);
\draw (0,1)--(5,1);
\draw plot [smooth] coordinates {(0,1) (-0.5,2) (6.5,2) (6,1)};
\draw plot [smooth] coordinates {(0,0) (-1.5,2.5)  (8.5,2.5) (8,0)};
\draw plot [smooth] coordinates {(2,-1) (1.75, -1.75) (7.25,-1.75) (7,-1)};
\end{tikzpicture}
\caption{A graph $\Gamma$ satisfying $\mathrm{PerfMatch}(\Gamma)=\mathrm{Tr}[f_2(d)h_{13}(c)g_{23}(b)g_{12}(a)]$\label{fig:examplegraph}}
\end{figure}
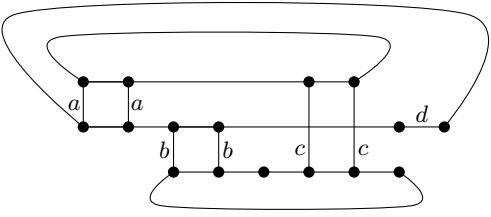
This gives the graph $\Gamma=(V,E,w)$ satisfying Eq.~\eqref{eq:gamma}, see Fig.~\ref{fig:examplegraph} for an example. Now let us verify the properties claimed above. Each gate contributes either $2,4$, or $6$ vertices, so the total number of vertices is $|V|=O(J)$. Each gadget is incident to $\leq 4$ external edges and contains $\leq 6$ internal edges, so the total number of edges is $|E|\leq 10J=O(J)$. Moreover, each external edge has weight $1$ and each edge inside a gadget has weight at most 2 (due to the restriction on $t$ in Eq.~\eqref{eq:Gndef}) , so the maximum edge weight is $w_{max}\leq 2$. To complete the proof, in the remainder of the paper we establish the following bound on $q(|V|)$:
\begin{equation}
\frac{\mathrm{NearPerfMatch}(\Gamma)}{\mathrm{PerfMatch}(\Gamma)}\leq O(|V|^2)=O(J^2).
\label{eq:polybound}
\end{equation}
\label{lem:ratiobound}
For any two distinct vertices $u,v\in V$ let $M_{u,v}(\Gamma)$ be the set of all nearly perfect matchings in which vertices $u$ and $v$ are unmatched. Define 
\[
\Omega_{u,v}(\Gamma)=\sum_{m\in {M_{u,v} (\Gamma)}}\prod_{e\in m} w(e).
\]
Since the total number of pairs $u,v$ of vertices in $\Gamma$ is $O(J^2)$, to prove Eq.~\eqref{eq:polybound} it suffices to show 
\begin{equation}
\frac{\Omega_{u,v}(\Gamma)}{\mathrm{PerfMatch}(\Gamma)}\leq O(1) 
\label{eq:suffcond}
\end{equation}
for all such pairs $u,v$. Instead of considering nearly perfect matchings of $\Gamma$ in which vertices $u,v$ are unmatched, it will be convenient to consider perfect matchings of a modified graph $\Gamma'_{u,v}$ obtained from $\Gamma$ by adding two new \textit{dangling edges}--that is, we add two new vertices labeled $u_0$ and $v_0$ and two new edges $(u,u_0)$ and $(v,v_0)$, each of weight one. Then any perfect matching of $\Gamma'_{u,v}$ contains both of these new edges  and $\Omega_{u,v}(\Gamma)=\mathrm{PerfMatch}(\Gamma'_{u,v})$. 

Let us now understand the effect of adding one or two dangling edges to one of the gadgets from Fig.~\ref{fig:gadgets}. If we attach a dangling edge to the input (resp. output) vertex of the gadget in Fig.~\ref{fig:gadgets}(a) then it is easy to check that the resulting graph gadget implements a gate $|1\rangle\langle0|=f(t)|1\rangle\langle0|$ (resp. $|0\rangle\langle 1|=|0\rangle\langle 1|f(t)$). If we attach dangling edges to both vertices it implements $|0\rangle\langle 0|=|0\rangle\langle1|f(t)|1\rangle\langle 0|$.

There are four possible gadgets obtained by attaching a single dangling edge to the gadget for $g(t)$ from Fig.~\ref{fig:gadgets}(b),  each equivalent to the one shown in Fig.~\ref{fig:gdangling} (a), up to a relabeling of external edges. These four gadgets implement gates
\begin{align}
&g(t)(I\otimes |1\rangle \langle 0|), \quad g(t)(|1\rangle \langle 0|\otimes I), \nonumber\\
&(I\otimes|0\rangle\langle 1|)g(t) ,\quad (|0\rangle\langle 1|\otimes I)g(t).
\label{eq:g1}
\end{align}
There are ${{4}\choose {2}}=6$ gadgets obtained by attaching two dangling edges to the gadget for $g(t)$ from Fig.~\ref{fig:gadgets}(b), each equivalent (up to relabeling of external edges) to one of the gadgets shown in Figs.~\ref{fig:gdangling} (b),(c), or (d). These $6$ gadgets implement gates
\begin{align}
&g(t)|11\rangle \langle 00|,\; (|0\rangle \langle 1|\otimes I)g(t)(|1\rangle \langle 0|\otimes I),\nonumber\\
&(I\otimes |0\rangle \langle 1|)g(t)(|1\rangle \langle 0|\otimes I),\; (|0\rangle \langle 1|\otimes I)g(t)(I\otimes |1\rangle \langle 0|),\nonumber\\
&(I\otimes |0\rangle \langle 1|)g(t)(I\otimes |1\rangle \langle 0|), \; |00 \rangle \langle 11|g(t).
\label{eq:g2}
\end{align}

Finally, consider adding one or two dangling edges to the gadget for $h(t)$ from Fig.~\ref{fig:gadgets}(c). One can again confirm by a direct (although tedious) inspection that the resulting collection of gadgets implement gates of the form
\begin{align}
P_ah(t),& \quad  h(t)P_a \qquad \qquad  \qquad \quad\;\; \text{(one dangling edge)} \nonumber\\
O_b h(t)P_a,& \quad O_bP_a h(t), \quad h(t)O_bP_a \quad \text{(two dangling edges)}\nonumber
\end{align}
where $a,b\in \{1,2\}$ are qubit indices and $O,P$ are single qubit operators from the set $\{|0\rangle \langle 1|,|1\rangle\langle 0|\}$.

In summary, adding a dangling edge to one of the gadgets from Fig.~\ref{fig:gadgets} modifies the gate implemented by the gadget by multiplying (either on the left or right) by an operator $|1\rangle \langle 0|$ or $|0\rangle\langle 1|$ (acting on one of the qubits). Likewise adding two dangling edges has the effect of applying two such multiplications.  Therefore 
\[
\Omega_{u,v}(\Gamma)=\trace{\left( G_J  \cdots G_j O_{b} G_{j-1} \cdots  G_i P_{a} G_{i-1} \cdots G_2 G_1\right)}
\]
for some qubit indices $a,b\in \{1,2,\ldots,n\}$, $1\leq i\leq j\leq J$, and $O,P\in  \{|0\rangle\langle 1|,|1\rangle \langle 0|\}$. Using the fact that all entries of the matrices $G_1,\ldots G_J$ are non-negative we infer
\begin{equation}
\Omega_{u,v}(\Gamma)\leq \trace{(B X_b A X_a)}
\label{eq:omegabound}
\end{equation}
where $X=|0\rangle \langle 1|+|1\rangle\langle 0|$, $A= G_{j-1} \cdots  G_{i+1} G_i$, and $B=G_{i-1} \cdots G_2 G_1G_J G_{J-1} \cdots G_j$.

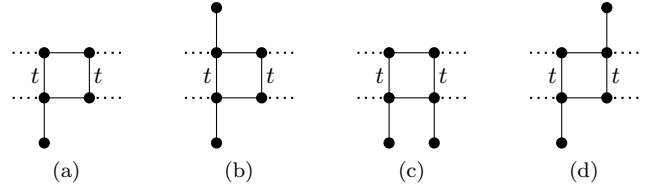
\begin{figure}
\centering
\subfloat[]{
\begin{tikzpicture}[scale=0.6]
\node at (0,-0.1){};
\draw node at (-0.2,0.5){$t$};
\draw node at (1.2,0.5){$t$};
\draw (0,0)--(0,1)--(1,1)--(1,0)--(0,0);
\draw (0,0)--(0,-1);
\draw node[circle,fill=black, inner sep=1.5pt] at (0,-1)(){};
\draw node[circle,fill=black, inner sep=1.5pt] at (0,0)(){};
\draw node[circle,fill=black, inner sep=1.5pt] at (0,1)(){};
\draw node[circle,fill=black, inner sep=1.5pt] at (1,0)(){};
\draw node[circle,fill=black, inner sep=1.5pt] at (1,1)(){};
\draw[thick,dotted] (-0.7,0)--(0,0);
\draw[thick,dotted] (1.7,0)--(1,0);
\draw[thick,dotted] (1,1)--(1.7,1);
\draw[thick,dotted] (0,1)--(-0.7,1);
\end{tikzpicture}
}
\hspace{0.5cm}
\subfloat[]{
\begin{tikzpicture}[scale=0.6]
\node at (0,-0.1){};
\draw node at (-0.2,0.5){$t$};
\draw node at (1.2,0.5){$t$};
\draw (0,0)--(0,1)--(1,1)--(1,0)--(0,0);
\draw (0,0)--(0,-1);
\draw(0,1)--(0,2);
\draw node[circle,fill=black, inner sep=1.5pt] at (0,2)(){};
\draw node[circle,fill=black, inner sep=1.5pt] at (0,-1)(){};
\draw node[circle,fill=black, inner sep=1.5pt] at (0,0)(){};
\draw node[circle,fill=black, inner sep=1.5pt] at (0,1)(){};
\draw node[circle,fill=black, inner sep=1.5pt] at (1,0)(){};
\draw node[circle,fill=black, inner sep=1.5pt] at (1,1)(){};
\draw[thick,dotted] (-0.7,0)--(0,0);
\draw[thick,dotted] (1.7,0)--(1,0);
\draw[thick,dotted] (1,1)--(1.7,1);
\draw[thick,dotted] (0,1)--(-0.7,1);
\end{tikzpicture}
}
\hspace{0.5cm}
\subfloat[]{
\begin{tikzpicture}[scale=0.6]
\node at (0,-0.1){};
\draw node at (-0.2,0.5){$t$};
\draw node at (1.2,0.5){$t$};
\draw (0,0)--(0,1)--(1,1)--(1,0)--(0,0);
\draw (0,0)--(0,-1);
\draw (1,0)--(1,-1);
\draw node[circle,fill=black, inner sep=1.5pt] at (1,-1)(){};
\draw node[circle,fill=black, inner sep=1.5pt] at (0,-1)(){};
\draw node[circle,fill=black, inner sep=1.5pt] at (0,0)(){};
\draw node[circle,fill=black, inner sep=1.5pt] at (0,1)(){};
\draw node[circle,fill=black, inner sep=1.5pt] at (1,0)(){};
\draw node[circle,fill=black, inner sep=1.5pt] at (1,1)(){};
\draw[thick,dotted] (-0.7,0)--(0,0);
\draw[thick,dotted] (1.7,0)--(1,0);
\draw[thick,dotted] (1,1)--(1.7,1);
\draw[thick,dotted] (0,1)--(-0.7,1);
\end{tikzpicture}
}
\hspace{0.5cm}
\subfloat[]{
\begin{tikzpicture}[scale=0.6]
\node at (0,-0.1){};
\draw node at (-0.2,0.5){$t$};
\draw node at (1.2,0.5){$t$};
\draw (0,0)--(0,1)--(1,1)--(1,0)--(0,0);
\draw (0,0)--(0,-1);
\draw (1,1)--(1,2);
\draw node[circle,fill=black, inner sep=1.5pt] at (1,2)(){};
\draw node[circle,fill=black, inner sep=1.5pt] at (0,-1)(){};
\draw node[circle,fill=black, inner sep=1.5pt] at (0,0)(){};
\draw node[circle,fill=black, inner sep=1.5pt] at (0,1)(){};
\draw node[circle,fill=black, inner sep=1.5pt] at (1,0)(){};
\draw node[circle,fill=black, inner sep=1.5pt] at (1,1)(){};
\draw[thick,dotted] (-0.7,0)--(0,0);
\draw[thick,dotted] (1.7,0)--(1,0);
\draw[thick,dotted] (1,1)--(1.7,1);
\draw[thick,dotted] (0,1)--(-0.7,1);
\end{tikzpicture}
}
\caption{Gadgets obtained from Fig~\ref{fig:gadgets}(b) by adding one or two dangling edges.  \label{fig:gdangling}}
\end{figure}
Let $s_k(A)$  be the $k$-th largest singular value of a matrix $A$
and $D=2^n$ be the dimension of the Hilbert space. 
It is known that for any matrices $A,B$  one has 
\begin{equation}
\label{fact}
\max_{U,V} |\trace{(BUAV)}| = \sum_{k=1}^D s_k(A) s_k(B),
\end{equation}
where the maximum is over all unitary matrices $U,V$.
Since $X_a$ and $X_b$ are unitary, Eq.~\eqref{eq:omegabound} implies
\begin{equation}
\Omega_{u,v}(\Gamma) \le \sum_{k=1}^D s_k(A) s_k(B).
\label{eq:near4}
\end{equation}
Define $\beta_a\equiv \beta(j-i)J^{-1}$ and $\beta_b \equiv \beta-\beta_a$. From Eq.~(\ref{eq:Gpartialprod}) one infers that $A=L e^{-\beta_a H + W} R$ and $B=L' e^{-\beta_b H + W'} R'$, where $W,W'$ are hermitian operators with norm at most $\epsilon$ and $L,R,L',R'$
are some operators with norm at most $2$. 
Then 
\[
s_k(A)\le \| L\|\cdot \|R\| \cdot s_k(e^{-\beta_a H + W}) \le O(1) e^{-\beta_a E_k},
\]
where $E_k$ is the $k$-th smallest eigenvalue of $H$. 
By the same argument, $s_k(B)\le O(1)e^{-\beta_b E_k}$. Substituting this into Eq.~\eqref{eq:near4} and using $\beta_a+\beta_b=\beta$ gives
\[
\Omega_{u,v}(\Gamma) \le O(1)\sum_{k=1}^{D} e^{-\beta E_k}=O(1)Z(\beta,H)\leq O(1)\mathcal{Z}_J
\]
where we used Eq.~\eqref{eq:Zmapprox}.  Recalling that $\mathcal{Z}_J=\mathrm{PerfMatch}(\Gamma)$ (cf. Eq.~\eqref{eq:gamma}) we arrive at Eq.~\eqref{eq:suffcond}. This completes the proof of Theorem \ref{thm:main}.

To conclude, we have shown that a large class of stoquastic Hamiltonians with ferromagnetic spin-spin interactions can be simulated classically in polynomial time by Monte Carlo algorithms. An interesting open question is whether recent extensions of Jerrum and Sinclair's techniques Refs.~\cite{mcquillan2013,huang2016} can be used to obtain efficient algorithms for a wider class of Hamiltonians beyond those of XY-type defined in Eq.~(\ref{eq:ham}).

\textit{Acknowledgments}: We acknowledge support from the IBM Research Frontiers Institute.

\bibliography{XYbib}


\onecolumngrid
\appendix
\section{Proof of Theorem \ref{thm:countpm}}
In this Appendix we prove Theorem \ref{thm:countpm}, following Section 5 of Ref.~\cite{JS89} closely with a few small modifications. 

We first introduce some additional notation. Throughout this Section $\Gamma=(V,E,w)$ is a graph with $|V|=2N$ vertices and positive edge weights $w(e)>0$ for all $e\in E$. Recall from the main text that we write $M_k (\Gamma)$ for the set of all matchings of $\Gamma$ containing exactly $k$ edges. We also
define the set of all matchings $M_{*}(\Gamma)=\bigcup_{k=0,1,\ldots,N} M_k(\Gamma)$. Define a positive weight function on matchings
\[
W(\Gamma,M)=\prod_{e\in M} w(e)\qquad \qquad M\in M_{*}(\Gamma)
\]
and weighted sums
\[
Z_k(\Gamma)=\sum_{M\in M_k(\Gamma)}W(\Gamma,M) \qquad \qquad Z(\Gamma)=\sum_{M\in M_*(\Gamma)} W(\Gamma,M).
\]
Here $\mathrm{PerfMatch}(\Gamma)=Z_N(\Gamma)$ and $\mathrm{NearPerfMatch}(\Gamma)=Z_{N-1}(\Gamma)$.  

In the following we say that $X$ approximates $Y$ within ratio $R$ iff 
\[
YR^{-1}\leq X\leq YR.
\]

We use the following theorem which is Corollary 4.3 of Ref.~\cite{JS89}. While Ref.~\cite{JS89} does not explicitly state the runtime bound, it is implicit in the proof.
\begin{theorem}[\textbf{Approximate sampler }\cite{JS89}]
There exists a classical probabilistic algorithm $\mathcal{A}(\Gamma,\epsilon)$ which takes as input a graph $\Gamma=(V,E,w)$ and a precision parameter $\epsilon>0$, and outputs a matching $M\in M_{*}(\Gamma)$ according to a probability distribution $P$. Moreover, for each $M\in M_{*}(\Gamma)$ the probability $P(M)$ approximates $\frac{W(\Gamma,M)}{Z(\Gamma)}$ within ratio $1+\epsilon$. The runtime of the algorithm is 
\begin{equation}
O\left(|E|^3|V|w_{max}^4 \log(w_{max}/w_{min})+|E|^2w_{max}^4 \log(\epsilon^{-1})\right).
\label{eq:sampler}
\end{equation}
\label{thm:samplealg}
\end{theorem}
The following ``log-concavity'' theorem was proven by Heilmann and Lieb in Ref. \cite{HL72} (a different proof for graphs with uniform edge weights is given in Ref. \cite{JS89}).

\begin{theorem}[Theorem 7.1 of Ref.~\cite{HL72}]
\[
Z_k(\Gamma)^2\geq Z_{k-1}(\Gamma)Z_{k+1}(\Gamma).
\]
\end{theorem}
As a direct consequence of log-concavity we obtain:
\begin{corol}
\[
\frac{1}{\sum_{e\in E}{w(e)}}=\frac{Z_{0}(\Gamma)}{Z_1(\Gamma)}\leq \frac{Z_{1}(\Gamma)}{Z_{2}(\Gamma)}\leq \ldots \leq  \frac{Z_{N-1}(\Gamma)}{Z_N(\Gamma)}.
\]
\label{cor:monotone}
\end{corol}
We are now going to show that, given $\Gamma$ and an integer $1\leq k\leq N$, we can modify the edge weights such that that the probability of matchings containing $k$ or $k+1$ edges is non-negligible. We shall write $\Gamma(\alpha)$ for the graph obtained from $\Gamma$ by multiplying all edge weights by $\alpha$, i.e.,  $\Gamma(\alpha)=(V,E,w')$ with $w'(e)=w(e)\alpha$. Define
\[
P_k(\alpha)=Z_k(\Gamma(\alpha))/Z(\Gamma(\alpha)).
\]
\begin{lemma}
Let an integer $1\leq k\leq N$ and $\Gamma=(V,E,w)$ be given. Suppose $\frac{Z_{N-1}(\Gamma)}{Z_N(\Gamma)}\leq q(N)$ where $q$ is a polynomial. Furthermore, suppose that $\alpha_k>0$ approximates $\frac{Z_{k-1}(\Gamma)}{Z_{k}(\Gamma)}$ within ratio $1+\frac{\epsilon}{2N}$ for some $\epsilon\in (0,1)$. Then 
\begin{equation}
P_k(\alpha_k)=\Omega(N^{-1}) \qquad \text{and} \qquad P_{k+1}(\alpha_k)=\Omega\left(N^{-1}|E|^{-1}q(N)^{-1}w_{max}^{-1}\right).
\label{eq:pklowerb}
\end{equation}
\label{lem:alphaweighted}
\end{lemma}
\begin{proof}
In the following for ease of notation we write $\alpha=\alpha_k$. First let $i\geq k$ and note that 
\[
\frac{P_{k}(\alpha)}{P_{i}(\alpha)}=\frac{Z_k(\Gamma)}{Z_{i}(\Gamma)}\alpha^{k-i}=\alpha^{k-i}\prod_{j=k}^{i-1}\frac{Z_j(\Gamma)}{Z_{j+1}(\Gamma)}\geq \alpha^{k-i}\left(\frac{Z_{k-1}(\Gamma)}{Z_{k}(\Gamma)}\right)^{i-k}\geq \left(1+\frac{\epsilon}{2N}\right)^{-N}\geq \frac{1}{2}.
\]
where in the third inequality we used Corollary \ref{cor:monotone} and in the last inequality we used the fact that $(1+\epsilon/2N)^N\leq 1+\epsilon\leq 2$. A symmetric argument establishes the same bound in the case $i<k$, i.e.,  $P_{k}(\alpha) \geq P_{i}(\alpha)/2$ for all $i$. Since $\sum_{i=0}^{N}P_i(\alpha)=1$ this implies $P_{k}(\alpha)\geq (2N+2)^{-1}$ which establishes the first claim in Eq.~\eqref{eq:pklowerb}.
We also have
\begin{equation}
P_{k+1}(\alpha)=\alpha \frac{Z_{k+1}(\Gamma)}{Z_k(\Gamma)} P_{k}(\alpha)
\geq \left(1+\frac{\epsilon}{2N}\right)^{-1}\frac{Z_{k-1}(\Gamma)}{Z_{k}(\Gamma)} \frac{Z_{k+1}(\Gamma)}{Z_k(\Gamma)} P_{k}(\alpha)
\geq  \frac{1}{2q(N)\sum_{e\in E} w(e)}P_k(\alpha).
\end{equation}
where in the last inequality we used Corollary \ref{cor:monotone} twice, along with the fact that $\frac{Z_{N-1}(\Gamma)}{Z_N(\Gamma)}\leq q(N)$. Substituting $\sum_{e\in E}w(e)\leq |E|w_{max}$ and $P_k(\alpha)=\Omega(N^{-1})$ gives the second claim in Eq.~\eqref{eq:pklowerb}.
\end{proof}

In the remainder of this Section we describe how the following algorithm, denoted $\mathcal{B}$, can be used to provide the randomized approximation scheme claimed in Theorem \ref{thm:countpm}. The algorithm takes as input a graph $\Gamma$, a positive integer $T$, a polynomial $q$ and a precision parameter $0<\delta<1$.
\begin{algorithm}[H]
\caption{$\mathcal{B}(\Gamma,T,q,\delta)$}
\label{alg:pm}
\begin{algorithmic}[1]
\State $\alpha_1 \gets \left(\sum_{e\in E} w(e)\right)^{-1}$ \Comment{Set $\alpha_1=\frac{Z_{0}(\Gamma)}{Z_{1}(\Gamma)}$}
\State $\Pi \gets \sum_{e\in E} w(e)$.    \Comment{ Set $\Pi=Z_1(\Gamma)$}
\For{\texttt{$k=1$ to $N-1$ }}
	 	\If {$\alpha_k>2q(N)$ or $\alpha_k< \left(2\sum_{e\in E} w(e)\right)^{-1}$} 
		\Return 0 
		\EndIf
		\State {Make $T$ calls to $\mathcal{A}(\Gamma(\alpha_k),\delta)$, resulting in outputs $Y=\{y_1,\ldots,y_T\}\in M_{*}(\Gamma)$.}
		\State {$p_k \gets T^{-1} |Y\cap M_{k}(\Gamma)|$}
		\State {$p_{k+1} \gets T^{-1} |Y\cap M_{k+1}(\Gamma)|$}
		\If {$p_k=0$ or $p_{k+1}=0$} 
		 \Return 0 
		\EndIf
		\State {$\alpha_{k+1} \gets \alpha_k p_k/p_{k+1}$} \Comment{$\alpha_{k+1}$ is our estimate of $\frac{Z_{k}(\Gamma)}{Z_{k+1}(\Gamma)}$}
		\State {$\Pi \gets \Pi/\alpha_{k+1}$} \Comment{$\Pi$ is our estimate of $Z_{k+1}(\Gamma)$}
\EndFor
\State{\Return $\Pi$}
\end{algorithmic}
\end{algorithm}
\begin{theorem}
Let $q$ be a polynomial and let $\epsilon>0$ be given. Let $\Gamma=(V,E,w)$ satisfy $\frac{Z_{N-1}(\Gamma)}{Z_N(\Gamma)}\leq q(N)$. One can choose  
\[
T=\tilde{\Theta}\left(\epsilon^{-2} N^4 |E|^2 w_{max}^2q(N)^2\right) \qquad \text{and}\qquad \delta=\Theta(\epsilon N^{-1})
\]
such that, with probability at least $3/4$,  the output of algorithm $\mathcal{B}(\Gamma,T,q,\delta)$ approximates $Z_N(\Gamma)$ within ratio $1+\epsilon$.
\label{thm:Balg}
\end{theorem}
\begin{proof}
 For each $1\leq k\leq N$ let $\mathcal{E}_k$ denote the event that the algorithm $\mathcal{B}(\Gamma,T,q,\delta)$ assigns a value to variable $\alpha_k$ before terminating and that this value approximates $\frac{Z_{k-1}(\Gamma)}{Z_{k}(\Gamma)}$ within ratio $1+\frac{\epsilon}{2N}$.  We shall prove inductively that
\begin{equation}
\mathrm{Pr}\left[\mathcal{E}_k\right]\geq \left(1-\frac{1}{4N^2}\right)^k \qquad \qquad  1\leq k\leq N.
\label{eq:ek}
\end{equation}
The theorem then follows directly from Eq.~\eqref{eq:ek}. Let $X$ denote the output of the algorithm. If events $\mathcal{E}_1,\mathcal{E}_2,\ldots \mathcal{E}_N$ all occur then $X=(\alpha_1\alpha_2\ldots \alpha_N)^{-1}$, and 
\[
\mathrm{Pr}\left[X \text{ approximates $Z_N(\Gamma)$ within ratio $(1+\epsilon/2N)^N$}\right]\geq \left(1-\frac{1}{4N^2}\right)^{\sum_{k=1}^{N} k}.
\]
Noting that $(1+\frac{\epsilon}{2N})^N\leq 1+\epsilon$ and that $\left(1-\frac{1}{4N^2}\right)^{\sum_{k=1}^{N} k}\geq \left(1-\frac{1}{4N^2}\right)^{N^2}\geq 3/4$ then completes the proof. 

It remains to establish Eq.~\eqref{eq:ek}. It holds trivially for $k=1$ since $\alpha_1=\left(\sum_{e\in E}w(e)\right)^{-1}=Z_0(\Gamma)/Z_1(\Gamma)$. For the inductive step let us suppose that Eq.~\eqref{eq:ek} holds for $k$. Then 
\[
\mathrm{Pr}[\mathcal{E}_{k+1}]\geq \mathrm{Pr}[\mathcal{E}_{k+1} | \mathcal{E}_{k}]\mathrm{Pr}[\mathcal{E}_{k}]\geq \left(1-\frac{1}{4N^2}\right)^{k}\mathrm{Pr}[\mathcal{E}_{k+1} | \mathcal{E}_{k}].
\]
To complete the proof it suffices to show that the conditional probability above satisfies 
\begin{equation}
\mathrm{Pr}[\mathcal{E}_{k+1} | \mathcal{E}_{k}]\geq \left(1-\frac{1}{4N^2}\right).
\label{eq:condprob}
\end{equation}
So now suppose that event $\mathcal{E}_{k}$ has occured. Let us examine what happens during the $k$th iteration of the for loop in the algorithm. By our inductive hypothesis, when the algorithm reaches the $k$th iteration of line 3 we have
\begin{equation}
\left(1+\frac{\epsilon}{2N}\right)^{-1} \frac{Z_{k-1}(\Gamma)}{Z_{k}(\Gamma)}\leq \alpha_{k} \leq \left(1+\frac{\epsilon}{2N}\right) \frac{Z_{k-1}(\Gamma)}{Z_{k}(\Gamma)}
\label{eq:induct}
\end{equation}
and since $ \left(\sum_{e\in E} w(e)\right)^{-1} \leq \frac{Z_{k-1}(\Gamma)}{Z_{k}(\Gamma)}\leq q(N)$ (by Corollary \ref{cor:monotone}) we see that the algorithm continues past line 4 without terminating. Let us now consider the values $p_k$ and $p_{k+1}$ which are subsequently assigned in lines 7 and 8. Both of these quantities are averages of i.i.d $0/1$-valued random variables:
\[
p_{k}=\frac{1}{T}\sum_{i=1}^{T} \mathbb{I}_k (y_i) \qquad  p_{k+1}=\frac{1}{T}\sum_{i=1}^{T} \mathbb{I}_{k+1} (y_i)  \qquad \mathbb{I}_j(M)=\begin{cases}1, & M\in M_j(\Gamma)\\0, &\text{otherwise.}\end{cases}
\]
Here each $y_i$ is drawn from the output of $\mathcal{A}(\Gamma(\alpha_k),\delta)$. Applying Theorem \ref{thm:samplealg} we see that 
\begin{equation}
\mathbb{E}[p_j] \text{ approximates } P_j(\alpha_k) \text{ within ratio } (1+\delta) \qquad j=k,k+1.
\label{eq:epj}
\end{equation}
Applying Hoeffding's inequality we get
\begin{equation}
\mathrm{Pr}\left[|p_j-\mathbb{E}[p_j]|\geq \mathbb{E}[p_j] \gamma\right]\leq 2e^{-2T(\mathbb{E}[p_j])^2 \gamma^2}\leq 2e^{-T(P_j(\alpha_k))^2 \gamma^2/2} \qquad j=k,k+1.
\label{eq:hoef}
\end{equation}
where in the last inequality we used the fact that Eq.~\eqref{eq:epj} implies $\mathbb{E}[p_j]\geq P_j(\alpha)/2$. Using Eq.~\eqref{eq:induct} and applying Lemma \ref{lem:alphaweighted} we see that $P_k(\alpha_k)$ and $P_{k+1}(\alpha_k)$ are bounded as in Eq.\eqref{eq:pklowerb}. Thus by choosing 
\[
T=\Theta\left(\frac{\log(N)}{(P_{k+1}(\alpha_k))^2\gamma^2}\right)=\tilde{\Theta}\left(\gamma^{-2} N^2 |E|^2 w_{max}^2q(N)^2\right),
\]
we can ensure that the right-hand side of Eq.~\eqref{eq:hoef} is at most $\frac{1}{8N^2}$ and therefore, with probability at least $1-\frac{1}{4N^2}$ we have:
\begin{equation}
\frac{p_k}{p_{k+1}} \text{ approximates } \frac{\mathbb{E}[p_k]}{\mathbb{E}[p_{k+1}]} \text{ within ratio } \frac{1+\gamma}{1-\gamma}.
\label{eq:pkratio}
\end{equation}
To complete the proof we now show that, for suitably chosen $\gamma,\delta$, Eqs.~(\ref{eq:epj},\ref{eq:pkratio}) together imply that event $\mathcal{E}_{k+1}$ occurs. Since Eq.~(\ref{eq:pkratio}) was shown to hold with probability at least $1-\frac{1}{4N^2}$ this proves Eq.~\eqref{eq:condprob}.

Putting together Eqs.~(\ref{eq:epj},\ref{eq:pkratio}) we get that
\begin{equation}
\frac{p_k}{p_{k+1}} \text{ approximates } \frac{P_k(\alpha_k)}{P_{k+1}(\alpha_k)} \text{ within ratio } R(\gamma,\delta)=\left(\frac{1+\gamma}{1-\gamma}\right)\left(1+\delta\right)^2.
\label{eq:pkratio2}
\end{equation}
Choose $\gamma=\epsilon/(c_1N)$ and $\delta=\epsilon/(c_2N)$ for absolute constants $c_1,c_2>0$ such that $R(\gamma,\delta)\leq 1+\frac{\epsilon}{2N}$. With this choice, and noting that $\alpha_k\frac{P_k(\alpha_k)}{P_{k+1}(\alpha_k)}=Z_{k}(\Gamma)/Z_{k+1}(\Gamma)$, we see that Eq.~\eqref{eq:pkratio2} implies that  $\alpha_{k+1}=\alpha_k p_k/p_{k+1}$ approximates $Z_{k}(\Gamma)/Z_{k+1}(\Gamma)$ within ratio $1+\frac{\epsilon}{2N}$. In other words event $\mathcal{E}_{k+1}$ occurs. This completes the proof.
\end{proof}
Finally we now complete the proof of Theorem \ref{thm:countpm}.
\begin{proof}[Proof of Theorem \ref{thm:countpm}]
Theorem \ref{thm:Balg} states that the output $X$ of the algorithm $\mathcal{B}(\Gamma,T,q,\delta)$ satisfies
\[
(1+\epsilon)^{-1}Z_N(\Gamma)\leq X\leq (1+\epsilon)Z_N(\Gamma)
\]
with probability at least $3/4$. Using the fact that $(1+\epsilon)^{-1}\geq (1-\epsilon)$ we see that the algorithm provides a randomized approximation scheme for $Z_N(\Gamma)=\mathrm{PerfMatch}(\Gamma)$ (provided that $T,\delta$ are chosen as specified in Theorem \ref{thm:Balg}). Now let us upper bound the runtime of the algorithm. Each time the subroutine $\mathcal{A}$ is called its graph argument has edge weights $w(e)\cdot \alpha_k$ for some $k$, which is always upper bounded by $w_{max} (2q(N))$ due to the condition in Line 4 of the algorithm. Using Eq.~\eqref{eq:sampler} with $w_{max}\rightarrow 2w_{max}q(N)$ the runtime of each such call is upper bounded by
\[
\tilde{O}\left(N|E|^3w_{max}^4 q(N)^4\right).
\]
Multiplying this by the maximum total number $NT$ of calls to $\mathcal{A}$ and substituting $|V|=2N$ we obtain the claimed runtime bound from theorem \ref{thm:countpm}.
\end{proof}


\section{Proof of Lemma \ref{lem:pathint}}
In this Appendix we prove Lemma \ref{lem:pathint}. We begin by stating bounds of the form Eqs.~(\ref{eq:gexp1}, \ref{eq:hexp1}). 
\begin{prop}
For $0<t<1$ we have
\begin{align}
g(t)&=e^{-t/2(Y\otimes Y-X\otimes X)+E(t)} \qquad \;\;\|E(t)\|\leq t^2  \label{eq:gexp2}\\
h(t)&=e^{-t/2(-Y\otimes Y-X\otimes X)+F(t)} \qquad \|F(t)\|\leq t^2.
\label{eq:hexp2}
\end{align}
\label{prop:ghbounds}
\end{prop}
We defer the (straightforward) proof of Proposition \ref{prop:ghbounds}. We shall also use the following bound.
\begin{lemma}
Let $H_1,H_2,\ldots, H_L$ be Hermitian operators with $\|H_i\|\leq \delta$ for all $i$, where $0\leq \delta L\leq 1/2$. Define 
\begin{equation}
C=e^{H_L/2}e^{H_{L-1}/2}\ldots e^{H_1/2}.
\label{eq:Qdef}
\end{equation}
Then there exists a Hermitian operator $\Delta$ such that
\[
CC^{\dagger}=e^{H_1+H_2+\ldots+H_L+\Delta} \qquad \text{and} \quad \|\Delta\|\leq 2\pi (\delta L)^3.
\]
\label{lem:Q}
\end{lemma}
\begin{proof}
Consider a time dependent Hamiltonian $H(t)$ with $t\in [-L,L]$ defined as follows:
\begin{equation}
\label{Ht}
H(t)=\left\{ \ba{rcl}
H_a/2 &\mbox{if}& a-1\le |t|<a \quad \mbox{for some $a=1,\ldots,L$} \\
0 && \mbox{otherwise.} \\
\ea\right.
\end{equation}
Let $U(t)$ be the solution of a differential equation
\begin{equation}
\label{Ut}
\frac{dU(t)}{dt}=H(t) U(t), \quad  -L\le t\le L.
\end{equation}
We choose initial conditions $U(-L)=I$. Note that $U(L)=CC^{\dagger}$. The Magnus expansion gives
\begin{equation}
\label{M1}
CC^{\dagger}=U(L)=\exp{[\Omega]}, \quad \Omega=\sum_{k=1}^\infty \Omega_k,
\end{equation}
where
\begin{equation}
\label{1st}
\Omega_1=\int_{-L}^L dt \, H(t)=H_1+H_2+\ldots + H_L,
\end{equation}
and
\begin{equation}
\label{2nd}
\Omega_2=\frac12 \int_{-L}^L dt \int_{-L}^t ds\,  [H(t),H(s)].
\end{equation}
The norm of the higher order terms can be bounded as 
\begin{equation}
\label{kth}
\| \Omega_k\| \le \pi \left( \int_{-L}^L \| H(t) \| dt \right)^k \le \pi (\delta L)^k.
\end{equation}
see page 29 of Ref.~\cite{Magnus}. Here in the last inequality we used the bound
$\|H(t)\|\le \delta/2$. Let us choose
\begin{equation}
\label{Vchoice}
\Delta=\Omega-\Omega_1.
\end{equation}
Since $\Omega_1$ and $\Omega$ are 
hermitian, we infer that $\Delta$ is hermitian. 
A direct inspection shows that $\Omega_2=0$. Therefore  Eq.~(\ref{kth}) gives
\begin{equation}
\label{Vnorm}
\|\Delta\| =\| \Omega-\Omega_1-\Omega_2\| \le \sum_{k=3}^\infty \|\Omega_k\|
\le \pi\sum_{k=3}^\infty  (\delta L)^k  \le \pi (\delta L)^3 \sum_{k=0}^\infty 2^{-k} =  2\pi (\delta L)^3.
\end{equation}
\end{proof}
We now use Eqs.~(\ref{eq:fexp1},\ref{eq:gexp2},\ref{eq:hexp2}) and Lemma \ref{lem:Q} to prove Lemma \ref{lem:pathint}.

\begin{proof}[Proof of Lemma \ref{lem:pathint}]
It will be convenient to rewrite the Hamiltonian  Eq.~\eqref{eq:ham} using coefficients $p_{ij}=(b_{ij}-c_{ij})/2$ and  $q_{ij}=(b_{ij}+c_{ij})/2$, i.e., 
\begin{equation}
H=\sum_{1\leq i<j\leq n} p_{ij}(-X_iX_j-Y_i Y_j)+\sum_{1\leq i<j\leq n}q_{ij}(-X_iX_j+Y_i Y_j)+\sum_{i=1}^{n} d_i(I+Z_i).
\label{eq:Hxy}
\end{equation}
Using the fact that $|c_{ij}|<b_{ij}\leq 1$ we see that $p_{ij},q_{ij}\in [0,1]$.

Let $0<\epsilon<1$ and $\beta>0$ be given. Let $r>2\beta $ be a positive integer which we will fix later. Define a rescaled Hamiltonian and rescaled coefficients
\begin{equation}
H'=\frac{\beta}{r} H  \qquad p'_{ij}=\frac{\beta}{r}p_{ij} \qquad q'_{ij}=\frac{\beta}{r}q_{ij}  \qquad d'_{i}= \frac{\beta}{r}d_{i}.
\label{eq:rescale}
\end{equation}
The rescaled coefficients satisfy 
\begin{equation}
0\leq p'_{ij},q'_{ij},|d'_i|\leq \frac{\beta}{r}< \frac{1}{2}.
\label{eq:scaledcoef}
\end{equation}
Consider a product
\begin{align}
C&=\prod_{1\leq i\leq n} f_{i}(e^{-d'_i}) \prod_{1\leq i<j\leq n} g_{ij}(q'_{ij})\prod_{1\leq i<j\leq n} h_{ij}(p'_{ij})\label{eq:productG}\\
&=\prod_{1\leq i\leq n} e^{-d'_i (I+Z_i)/2} \prod_{1\leq i<j\leq n} e^{-q'_{ij}/2(Y_iY_j-X_iX_j)+E_{ij}} \prod_{1\leq i<j\leq n}e^{-p'_{ij}/2(-Y_iY_j-X_iX_j)+F_{ij}}
\label{eq:prodxy}
\end{align}
where in the second line we used Eqs.~(\ref{eq:fexp1},\ref{eq:gexp2},\ref{eq:hexp2}). Here the Hermitian operators $E_{ij},F_{ij}$ satisfy
\begin{equation}
\|E_{ij}\|\leq (q'_{ij})^2\leq \frac{\beta^2}{r^2} \qquad  \|F_{ij}\|\leq (p'_{ij})^2\leq \frac{\beta^2}{r^2}.
\label{eq:efbound}
\end{equation}
The bounds Eq.~\eqref{eq:scaledcoef} and the fact that $e^{-d'_i}\leq e^{1/2}<2$ ensure that  Eq.~\eqref{eq:productG} is a product of $n^2$ gates from the set $\mathcal{G}$ defined in Eq.~\eqref{eq:Gndef}. 
Furthermore, Eq.~\eqref{eq:prodxy} has the form Eq.~\eqref{eq:Qdef} with $L=n^2$, and 
\begin{equation}
\|H_i\|\leq \max_{jk} \bigg\{|2d'_j|, \; 2q'_{jk}+2\|E_{jk}\|, \; 2p'_{jk}+2\|F_{jk}\|\bigg\}\leq \left(2\beta/r+2\beta^2/r^2\right)\leq \frac{3\beta}{r}
\label{eq:hbound}
\end{equation}
where we used Eq.~\eqref{eq:scaledcoef}. Applying Lemma \ref{lem:Q} and using Eqs.~(\ref{eq:Hxy},\ref{eq:rescale}) gives
\[
CC^{\dagger}=\exp\left[-H'+\sum_{1\leq i<j\leq n} (2E_{ij}+2F_{ij})+\Delta\right] \qquad \|\Delta\|\leq 2\pi\left(\frac{3\beta n^2}{r}\right)^3
\]
as long as our choice of $r$ satisfies 
\begin{equation}
6 \beta n^2 r^{-1}\leq 1
\label{eq:rcond}
\end{equation}
(which will be the case, see below). Using Eq.~\eqref{eq:efbound} and the triangle inequality gives
\begin{equation}
CC^{\dagger}=e^{-H'+D} \qquad \|D\|\leq \frac{2n^2\beta^2}{r^2}+2\pi\left(\frac{3\beta n^2}{r}\right)^3,
\label{eq:ccdag}
\end{equation}
and 
\begin{equation}
(CC^{\dagger})^r=e^{-rH'+rD}\equiv e^{-\beta H+Q} \qquad \|Q\|\leq\frac{2n^2\beta^2}{r}+\frac{2\pi3^3 \beta^3 n^6}{r^2}.
\label{eq:qq}
\end{equation}
Since $C$ is a product of $n^2$ gates from $\mathcal{G}$ the left hand side is a product of $J=2n^2r$ such gates. At the end of the proof we will choose $r$ to ensure that $\|Q\|\leq \epsilon/4$. 

Next consider a partial product of the form given on the left hand side of Eq.~\eqref{eq:Gpartialprod}. Since $CC^{\dagger}$ is a product of $2n^2$ gates and $G_JG_{J-1}\ldots G_1=(CC^{\dagger})^r$, we may write
\begin{equation}
G_jG_{j-1}\ldots G_i= L_{ij}(CC^{\dagger})^{K}R_{ij}
\label{eq:Gjtoi}
\end{equation}
where $K\geq 0$ and $R_{ij}$ and $L_{ij}$ are each products of at most $2n^2-1$ gates $\{G_t\}$, and furthermore $R_{1j}=L_{iJ}=I$. Since each gate is of the form $G_t=e^{H_t/2}$ where $H_i$ satisfies Eq.~\eqref{eq:hbound}, we have $\|G_t\|\leq e^{\frac{3\beta}{2r}}$, and thus
\begin{equation}
\|R_{ij}\|,\|L_{ij}\|\leq e^{3\frac{\beta}{2r} (2n^2-1)}\leq e^{1/2}\leq 2
\label{eq:RLbound}
\end{equation}
where we used Eq.~\eqref{eq:rcond}. Moreover, since the left-hand side of Eq.~\eqref{eq:Gjtoi} contains $j-i+1$ gates in total and $R_{ij},L_{ij}$ contain at most $2n^2-1$ gates each, we have
\[
0\leq (j-i+1)-2n^2K\leq 2(2n^2-1),
\]
and therefore
\begin{equation}
\left|K-\frac{(j-i+1)}{2n^2}\right|\leq 2.
\label{eq:kbound}
\end{equation}
Combining Eqs.~(\ref{eq:ccdag},\ref{eq:Gjtoi}) we obtain
\begin{equation}
G_jG_{j-1}\ldots G_i=L_{ij}e^{-KH'+KD}R_{ij},
\label{eq:Gjtoi2}
\end{equation}
where
\begin{equation}
\|KD\|\leq \|rD\|=\|Q\|
\label{eq:NDbound}
\end{equation}
and, using $H'=\beta/r H$ and Eq.~\eqref{eq:kbound},
\begin{equation}
\|KH'- \frac{(j-i+1)\beta}{2n^2r}H\|\leq \frac{2\beta}{r}\|H\|\leq 4n^2\beta/r.
\label{eq:NH}
\end{equation}
Combining Eqs.~(\ref{eq:Gjtoi2},\ref{eq:NDbound},\ref{eq:NH}) and using the fact that $J=2n^2r$ gives Eq.~(\ref{eq:Gpartialprod}) with 
\[
\|W_{ij}\|\leq \frac{4n^2\beta}{r}+\|Q\|\leq \frac{4n^2\beta}{r}+\frac{2n^2\beta^2}{r}+\frac{2\pi3^3 \beta^3 n^6}{r^2},
\]
where in the last inequality we used Eq.~\eqref{eq:qq}. Now choose $r=O(n^3\lceil \beta\rceil^2\epsilon^{-1})$ to make the right hand side at most $\epsilon/4$ and such that Eq.~\eqref{eq:rcond} is also satisfied. This gives Eqs.~(\ref{eq:Gprod},\ref{eq:Gpartialprod}) with $\|Q\|\leq \|W_{ij}\|\leq \epsilon/4$. Noting that $J=2n^2r=O(n^5(\beta^2+1)\epsilon^{-1})$ completes the proof.
\end{proof} 
Finally, we prove Proposition \ref{prop:ghbounds}.
\begin{proof}[Proof of Proposition \ref{prop:ghbounds}]
We have the equality
\begin{equation}
g(t)=\exp{\left[\frac{1}{2}R(t)(X\otimes X-Y\otimes Y)+\frac{t}{4} R(t)\left(Z\otimes I +I\otimes Z\right)\right]}
\label{eq:gexact}
\end{equation}
for all $t>0$, where 
\[
R(t)=\frac{1}{\sqrt{1+t^2/4}}\cosh^{-1}(1+t^2/2)=\frac{1}{\sqrt{1+t^2/4}}\log\left(1+t\sqrt{1+t^2/4}+t^2/2\right).
\]
Using a second order Taylor expansion about $t=0$ one can confirm that
\begin{equation}
|R(t)-t|\leq t^3/6 \qquad 0<t< 1.
\label{eq:taylor}
\end{equation}
Indeed, $R(t)$ has a continuous second derivative on $[0,1]$, and is thrice differentiable on the open interval $(0,1)$. Applying Taylor's theorem we obtain
\begin{equation}
R(t)=t+0\cdot t^2+\mathrm{Error} \qquad \qquad \left|\mathrm{Error}\right|\leq \frac{t^3}{3!}\max_{(0,1)}\left|\frac{d^{3}R}{dt^3}\right| \qquad \qquad 0<t<1.
\label{eq:lagrange}
\end{equation}
Here 
\[
\frac{d^{3}R}{dt^3}=\frac{44t^2-64}{(t^2+4)^3}+\frac{72t-12t^3}{(t^2+4)^{7/2}}\cosh^{-1}(1+t^2/2)  \qquad t>0.
\]
The first term is negative and has magnitude at most $1$ on the interval $(0,1)$ whereas the second term is nonnegative and has magnitude at most $\frac{72}{4^{7/2}}\cosh^{-1}(3/2)=0.54...$ on $(0,1)$. Therefore $\max_{(0,1)}\left|\frac{d^{3}R}{dt^3}\right|\leq 1$ and plugging into Eq.~\eqref{eq:lagrange} gives Eq.~\eqref{eq:taylor}.

From Eq.~\eqref{eq:taylor} and $\|X\otimes X-Y\otimes Y\|=\|\left(Z\otimes I +I\otimes Z\right)\|=2$ we obtain
\[
\left\|\frac{1}{2}(R(t)-t)(X\otimes X-Y\otimes Y)+\frac{t}{4}R(t)\left(Z\otimes I +I\otimes Z\right)\right\|\leq \frac{t^3}{6}+\frac{t}{2}\left(t+\frac{t^3}{6}\right)\leq t^2 \qquad \qquad 0<t<1.
\]
Using this bound in Eq.~\eqref{eq:gexact}  we arrive at Eq.~\eqref{eq:gexp2}.

Finally, from Eq.~\eqref{eq:twoqubitgates} we see that $h(t)=(I\otimes X)g(t)(I\otimes X)$ and therefore Eq.~\eqref{eq:gexp2} implies Eq.~\eqref{eq:hexp2} where $F(t)=(I\otimes X)E(t)(I\otimes X)$.
\end{proof}
\end{document}